\documentclass[submission,copyright,creativecommons,leqno]{eptcs}
\providecommand{\event}{NCL 22} 
\usepackage{underscore}           

\usepackage{amsmath}
\usepackage{amssymb}
\usepackage[retainorgcmds]{IEEEtrantools} 
\usepackage{verbatim} 
\usepackage{bussproofs} 
\usepackage{color}
\usepackage{amsthm}
\usepackage{float}

\theoremstyle{definition}
\newtheorem{theorem}{Theorem} 
\newtheorem{definition}{Definition}
\newtheorem{proposition}[definition]{Proposition}
\newtheorem{lemma}[definition]{Lemma}
\newtheorem{corollary}[definition]{Corollary}
\newtheorem{remark}[definition]{Remark}

\newcommand{\Not}{{\sim}}

\title{Another Combination of \\ Classical and Intuitionistic Conditionals\thanks{The research by Satoru Niki was supported by funding from German Academic Exchange Service (DAAD) and the European Research Council (ERC) under the European Union's Horizon 2020 research and innovation programme, grant agreement ERC-2020-ADG, 101018280, ConLog. The research by Hitoshi Omori has been supported by a Sofja Kovalevskaja Award of the Alexander von Humboldt-Foundation, funded by the German Ministry for Education and Research. We would like to thank the anonymous referees for their careful and helpful comments which helped in improving the presentation of the paper.}}
\author{
Satoru Niki \qquad\qquad Hitoshi Omori
\institute{Department of Philosophy I\\ 
Ruhr University Bochum\\
Bochum, Germany}
\email{Satoru.Niki@rub.de \quad\qquad Hitoshi.Omori@rub.de}
}
\def\titlerunning{Another Combination of Classical and Intuitionistic Conditionals}
\def\authorrunning{S. Niki \& H. Omori}
\begin{document}
\maketitle

\begin{abstract}
On the one hand, classical logic is an extremely successful theory, even if not being perfect. On the other hand, intuitionistic logic is, without a doubt, one of the most important non-classical logics. But, how can proponents of one logic view the other logic? In this paper, we focus on one of the directions, namely how classicists can view intuitionistic logic. To this end, we introduce an expansion of positive intuitionistic logic, both semantically and proof-theoretically, and establish soundness and strong completeness. Moreover, we discuss the interesting status of disjunction, and the possibility of combining classical logic and minimal logic. We also compare our system with the system of Caleiro and Ramos.
\end{abstract}

\section{Introduction}

Intuitionistic logic is, without a doubt, one of the most important non-classical logics. For the purpose of illustrating our motivation, let us refer to those who, for whatever reasons they may have, endorse intuitionistic logic and classical logic, as \emph{intuitionists} and \emph{classicists}, respectively. Then, we are interested in how one camp views the logic endorsed by the other camp. In other words, we are interested in the following two questions.

\smallskip

{\bf (Q1)} how do/can intuitionists view classical logic? 

{\bf (Q2)} how do/can classicists view intuitionistic logic?

\smallskip

\noindent In what follows, we will somewhat naively refer to intuitionistic logic and classical logic as $\vdash_{\bf IL}$ and $\vdash_{\bf CL}$, respectively, and also to intuitionistic conditional and classical conditional as $\to$ and $\supset$, respectively. Moreover, we will assume that the conditionals $\to$ and $\supset$ internalize the consequence relations $\vdash_{\bf IL}$ and $\vdash_{\bf CL}$, respectively. That is, we assume that the Deduction Theorem, namely $\Gamma, A\vdash_{\bf IL} B$ iff $\Gamma\vdash_{\bf IL} A{\to} B$ and $\Gamma, A\vdash_{\bf CL} B$ iff $\Gamma\vdash_{\bf CL} A{\supset} B$ both hold in the appropriate language.

Let us begin with {\bf (Q1)}. Given the basic assumption concerning the Deduction Theorem, how can intuitionists accommodate classical conditional? If intuitionists accept both that (i) everything provable for $\to$ is also provable for $\supset$ (i.e. if $\vdash_{\bf IL} A{\to} B$, then $\vdash_{\bf IL} A{\supset} B$) and that (ii) \emph{modus ponens} holds for $\supset$ (i.e. $A, A{\supset} B\vdash_{\bf IL} B$), then it follows that $\vdash_{\bf IL} A{\to} B$ iff $\vdash_{\bf IL} A{\supset} B$. This might be seen as a collapse, as observed in \cite{gabbay1996overview} by Dov Gabbay, and thus intuitionists may need to give up one of the assumptions in viewing classical conditional. One option is to give up \emph{modus ponens} with respect to $\supset$, and this is pursued by Dag Prawiz in \cite{Prawitz2015}, followed by Luiz Carlos Pereira and Ricardo Oscar Rodriguez in  \cite{pereira2017normalization} as well as by Elaine Pimentel, L.C. Pereira and Valeria de Paiva in \cite{pimentel2018,pimentel2019ecumenical}. Another option might be to rely on negative translations. This is the way that seems to be endorsed, for example, by Helmut Schwichtenberg  and Stanley S. Wainer (cf. \cite[pp.4-15]{schwichtenberg2011proofs}), but there are also some worries summarized with further references by Lloyd Humberstone (cf. \cite[p.305]{Humberstone2011connectives}).

For the other question, namely {\bf (Q2)}, this may seem to be fairly straightforward since there is Kripke semantics for intuitionistic logic. However, note that the classical conditional $\supset$ is \emph{not} present in the standard semantics as a primitive or defined connective. Moreover, a straightforward attempt of adding classical conditional, which is equivalent to the addition of classical negation, will face a difficulty. Indeed, as observed by Luis Fari{\~n}as del Cerro and Andreas Herzig in \cite[pp.93--94]{dCH1996} and Paqui Lucio in \cite{lucio2000structured}, the heredity condition will be in conflict with the classical conditional/negation. Accepting this conflict is the approach taken in \cite{dCH1996,lucio2000structured} and more recently by Masanobu Toyooka and Katsuhiko Sano in \cite{toyooka2021analytic}.
To the best of our knowledge, previous attempts so far in the literature are facing troubles, in many cases adding \emph{ad hoc} restrictions, such as the atoms being divided into classical ones and intuitionistic ones. Therefore, it seems to be difficult to conclude that these systems, designed for different goals, offer a way to view intuitionistic logic from classicists' perspective.

Based on these, our aim in this paper is to present a new combination of classical logic and intuitionistic logic
for the purpose of addressing {\bf (Q2)}. Given that there are worries raised for the absurdity constant among the constructivists' camp, we will focus on the positive fragment \emph{without} negation, although we can restore the absurdity constant, if desired. This will have a nice and surprising byproduct of allowing us to consider a combination of classical logic and minimal logic, a topic that seems to have never discussed so far in the literature. 

\section{Semantics and proof system}\label{sec:semanticsproofsystem}

The basic languages $\mathcal{L}$ and $\mathcal{L}^-$ consist of sets $\{ \land, \lor , \to \}$ and $\{ \land,  \to \}$, respectively, of propositional connectives and a countable set $\mathsf{Prop}$ of propositional variables which we denote by $p$, $q$, etc. Furthermore, we denote by $\mathsf{Form}$ and $\mathsf{Form}^-$ the set of formulas defined as usual in $\mathcal{L}$ and $\mathcal{L}^-$, respectively. We denote a formula of the given language by $A$, $B$, $C$, etc. and a set of formulas by $\Gamma$, $\Delta$, $\Sigma$, etc. We will use $A\leftrightarrow B$ as an abbreviation for $(A\to B)\land(B\to A)$.

Moreover, we consider a few expansions of both languages $\mathcal{L}$ and $\mathcal{L}^-$. The additional connectives will be made explicit by subscripts. The additional connectives in this paper include $\bot$, $\Box$, and $\supset$. For example, we refer to the language obtained by adding $\supset$ to $\mathcal{L}$ as $\mathcal{L}_\supset$, and the set of formulas as $\mathsf{Form}_\supset$.

Let us now state the semantics. 

\begin{definition}\label{def.model}
An {\bf S}-model for $\mathcal{L}_\supset$ is a quadruple $\langle g, W, \leq, V \rangle$, where $W$ is a set of states with $g\in W$ (the base state); $\leq$ is a reflexive and transitive relation on $W$ with $g$ being the least element; and $V: W\times \mathsf{Prop} \longrightarrow \{ 0 , 1 \}$ is an assignment of truth values to state-variable pairs with the condition that $V(w_1, p)=1$ and $w_1\leq w_2$ only if $V(w_2, p)=1$ for all $p\in \mathsf{Prop}$, all $w_1, w_2\in W$. Valuations $V$ are then extended to interpretations $I$ of state-formula pairs by the following conditions:

\vspace{-2mm}
\begin{itemize}
\setlength{\parskip}{0cm}
\setlength{\itemsep}{0cm}
\item $I(w,p)=V(w,p)$;
\item $I(w, A\land B)=1$ iff $I(w, A)=1$ and $I(w, B)=1$;
\item $I(w, A\lor B)=1$ iff $I(w, A)=1$ or $I(w, B)=1$;
\item $I(w, A{\supset} B)=1$ iff $I(g, A)\neq 1$ or $I(w, B)=1$;
\item $I(w, A{\to} B)=1$ iff for all $x\in W$: if $w\leq x$ and $I(x, A)=1$ then $I(x, B)=1$.
\end{itemize}
\vspace{-2mm}
Finally, semantic consequence is now defined as follows: $\Sigma \models A$ iff for all {\bf S}-models $\langle g, W, \leq, I \rangle$, $I(g, A)=1$ if $I(g, B)=1$ for all $B\in \Sigma$.
\end{definition}

\begin{remark}
Since we assume our meta-theory to be  classical, the clause for $\supset$ is equivalent to `if $I(g,A)=1$ then $I(w,B)=1$' as well as to `for all $x\in W$: if $w\leq x$ and $I(g, A)=1$ then $I(x, B)=1$'.
\end{remark}

Then, by induction on the complexity of formulas, we obtain the following.
\begin{lemma}[Persistence]
Let $\langle g, W, \leq, V \rangle$ be an {\bf S}-model for $\mathcal{L}_\supset$. Then, if $I(w_1, A){=}1$ and $w_1{\leq} w_2$ then $I(w_2, A){=}1$ for all $A\in \mathsf{Form}_\supset$ and for all $w_1, w_2\in W$.
\end{lemma}

\begin{remark}
Unlike the system considered by del Cerro and Herzig in \cite{dCH1996} in which the persistence condition is restricted, we have the \emph{full} persistence.
\end{remark}

Let us turn to relate $\supset$ to the notions of \emph{empirical negation} (cf. \cite{MikeEN,DeOmori2014,DeOmori2016}) and \emph{actuality operator} (cf. \cite{Humberstone2006,NikiOmoriAiMLsubmitted}), which are defined as follows.

\begin{definition}
Empirical negation $\Not$ and actuality operator $@$ are semantically defined by truth conditions $I(w, \Not A)=1$ iff $I(g, A)\neq 1$ and $I(w, @A)=1$ iff $I(g, A)=1$, respectively.
\end{definition}

One can easily observe that $A{\supset} B$ can be defined as $\Not A\lor B$, if empirical negation is available in the language. Note, however, that $\Not$ is not definable in our model, as observed in the next lemma. 

\begin{lemma}\label{lem:Not@def}
For any $\langle g,W,\leq\rangle$, there is an {\bf S}-model $\langle g,W,\leq,V\rangle$ such that for any formula $B[p]$ it is not the case that $I(w, B[p/A])=1$ iff $I(g, A)=0$ for all $w\in W$. Similarly,  For any $\langle g,W,\leq\rangle$ with $\#W\geq 2$ (i.e. more than one world), there is an {\bf S}-model $\langle g,W,\leq,V\rangle$ such that for any formula $B[p]$ it is not the case that $I(w, B[p/A])=1$ iff $I(g, A)=1$ for all $w\in W$.
\end{lemma}

\begin{proof}
For the former, consider the model in which $V(g, p)=1$ for all $p\in \mathsf{Prop}$. Then, we obtain $I(g, A)=1$ for all $A\in \mathsf{Form}_\supset$. On the other hand, if $\Not$ were definable in the model, then $I(g,\Not(p\to p))=0$, a contradiction. For the latter, consider an {\bf S}-model with $W=\{g,w\}$. We assign values {\bf 1}, {\bf i} and {\bf 0} to a formula $A$ when $I(g,A)=1$, $I(g,A)=0$ but $I(w,A)=1$ and $I(w,A)\neq 0$, respectively. We shall write $I(A)$ for such an assignment. Then we obtain the following truth table: 
\begin{center} 
$
\begin{tabular}{c|cccc}
$A \land B$ & $\mathbf{1}$ & $\mathbf{i}$  & $\mathbf{0}$ \\ \hline
$\mathbf{1}$ & $\mathbf{1}$ & $\mathbf{i}$ & $\mathbf{0}$\\
$\mathbf{i}$ & $\mathbf{i}$ & $\mathbf{i}$ & $\mathbf{0}$\\
$\mathbf{0}$ & $\mathbf{0}$ & $\mathbf{0}$  & $\mathbf{0}$\\
\end{tabular}
\quad
\begin{tabular}{c|cccc}
$A \lor B$ & $\mathbf{1}$ & $\mathbf{i}$  & $\mathbf{0}$ \\ \hline
$\mathbf{1}$ & $\mathbf{1}$ & $\mathbf{1}$ & $\mathbf{1}$\\
$\mathbf{i}$ & $\mathbf{1}$ & $\mathbf{i}$ & $\mathbf{i}$\\
$\mathbf{0}$ & $\mathbf{1}$ & $\mathbf{i}$ & $\mathbf{0}$
\end{tabular}
\quad
\begin{tabular}{c|cccc}
$A {\supset} B$ & $\mathbf{1}$ & $\mathbf{i}$  & $\mathbf{0}$ \\ \hline
$\mathbf{1}$ & $\mathbf{1}$ & $\mathbf{i}$ & $\mathbf{0}$\\
$\mathbf{i}$ & $\mathbf{1}$ & $\mathbf{1}$ & $\mathbf{1}$\\
$\mathbf{0}$ & $\mathbf{1}$ & $\mathbf{1}$ & $\mathbf{1}$
\end{tabular}
\quad
\begin{tabular}{c|cccc}
$A {\to} B$ & $\mathbf{1}$ & $\mathbf{i}$  & $\mathbf{0}$ \\ \hline
$\mathbf{1}$ & $\mathbf{1}$ & $\mathbf{i}$ & $\mathbf{0}$\\
$\mathbf{i}$ & $\mathbf{1}$ & $\mathbf{1}$ & $\mathbf{0}$\\
$\mathbf{0}$ & $\mathbf{1}$ & $\mathbf{1}$ & $\mathbf{1}$
\end{tabular}
$
\end{center}
Then, if $@$ is definable in {\bf S}-models with $\#W\geq 2$, then it is also definable in the above three-valued semantics. In particular, we need $I(@A)=\mathbf{0}$ when $I(A)=\mathbf{i}$. However, this is not possible since in the above three-valued semantics, we obtain that $I(A)\neq\mathbf{0}$ for all $A\in \mathsf{Form}_\supset$ if we consider models with $I(p)\neq \mathbf{0}$ (in other words, $\{ \mathbf{1}, \mathbf{i} \}$ is closed under the given four truth functions). 
\end{proof}

We now turn to the proof system.
\begin{definition}
The system {\bf S} consists of the following axiom schemata and a rule of inference:

\vspace{-3mm}

\noindent
\begin{minipage}{.48\textwidth}
\begin{gather}
A {\to} (B {\to} A) \tag{Ax1} \label{Ax1} \\
(A {\to} (B {\to} C)){\to} ((A {\to} B)  {\to} (A {\to} C)) \tag{Ax2} \label{Ax2} \\
(A \land B) {\to} A \tag{Ax3} \label{Ax3} \\
(A \land B) {\to} B \tag{Ax4} \label{Ax4}\\
(C{\to} A) {\to} ((C{\to} B) {\to} (C{\to} (A {\land} B))) \tag{Ax5} \label{Ax5}\\
A {\to} (A \lor B) \tag{Ax6} \label{Ax6} \\
B {\to} (A \lor B) \tag{Ax7} \label{Ax7}\\
(A {\to} C) {\to} ((B {\to} C) {\to} ((A {\lor} B) {\to} C)) \tag{Ax8} \label{Ax8}
\end{gather}
\end{minipage}
\begin{minipage}{.5\textwidth}
\begin{gather}
(A{\to} B){\supset} (A{\supset} B) \tag{AxM1} \label{AxM1} \\
(A {\supset} (B {\supset} C)){\to} ((A {\supset} B)  {\supset} (A {\supset} C)) \tag{AxM2} \label{AxM2} \\
(A{\supset}(B{\to}C)){\to}(B{\to}(A{\supset}C)) \tag{AxM3} \label{AxM3}\\
(A{\to}(B{\supset}C)){\to}(B{\supset}(A{\to}C)) \tag{AxM4} \label{AxM4}\\
((A{\supset} B){\supset}C){\to}((A{\supset C}){\to} C) \tag{AxM5} \label{AxM5}\\
(A{\supset}C){\to}((B{\supset}C){\to}((A{\lor} B){\supset}C)) \label{AxM6} \tag{AxM6}\\
\frac{\ A\quad A{\supset} B \ }{B} \label{MP} \tag{MP} 
\end{gather}
\end{minipage}

\medskip

\noindent Finally, we write $\Gamma \vdash A$ if there is a sequence of formulas $B_1, \dots, B_n, A$, $n\geq 0$, such that every formula in the sequence $B_1, \dots , B_n, A$ either (i) belongs to $\Gamma$; (ii) is an instance of an axiom of {\bf S}; (iii) is obtained by \eqref{MP} from formulas preceding it.
\end{definition}

\begin{proposition}\label{prop.inf}
The following rule and formulas are derivable in {\bf S}.
\vspace{-2mm}

\noindent 
\begin{minipage}{.49\textwidth}
\begin{gather}
\frac{\ A\quad A{\to}B\ }{B} \tag{MP2} \label{MP2}\\
A{\supset}((A{\supset}B){\to} B) \tag{MP3} \label{MP3} \\
A\to (B\supset A) \tag{Kmix} \label{Kmix} 
\end{gather}
\end{minipage}
\begin{minipage}{.5\textwidth}
\begin{gather}
A {\supset} (B {\supset} A) \tag{K$\supset$} \label{Ksup} \\
(A {\supset} (B {\supset} C)){\supset} ((A {\supset} B)  {\supset} (A {\supset} C)) \tag{S$\supset$} \label{Ssup}\\
A\lor (A\supset B) \tag{C} \label{AxC}
\end{gather}
\end{minipage}
\end{proposition}

\begin{proof}
\eqref{MP2} follows from \eqref{AxM1} and \eqref{MP}. Then \eqref{MP3} follows from $(A{\supset} B){\to}(A{\supset} B)$, \eqref{AxM4} and \eqref{MP2}. 
For \eqref{Kmix}, first we deduce $B{\supset}(A{\to} A)$ from $B{\to}(A{\to} A)$, \eqref{AxM1} and \eqref{MP}; then use \eqref{AxM3} and \eqref{MP2}. \eqref{Ksup} and \eqref{Ssup} follow from \eqref{AxM1}, \eqref{MP} and \eqref{Kmix} or \eqref{AxM2}, respectively. Finally, \eqref{AxC} is derived from \eqref{AxM5} with $(A{\supset} B){\supset}(A{\lor} (A{\supset} B))$ and $A{\supset} (A{\lor}(A{\supset} B))$, which follow from \eqref{AxM1}. 
\end{proof}

\eqref{Ksup}, \eqref{Ssup} and \eqref{AxC} show that $\supset$ does represent classical implication. Moreover, given that we have \eqref{Ksup} and \eqref{Ssup}, and that \eqref{MP} is the only rule of inference, we obtain the following Deduction theorem.

\begin{proposition}[Deduction theorem]\label{prop:DT}
For all $\Gamma \cup \{ A , B \} \subseteq \mathsf{Form}_\supset$, $\Gamma , A\vdash B$ iff $\Gamma \vdash A{\supset} B$.
\end{proposition}


\begin{remark}
Given that we are interested in {\bf (Q2)} of the introduction, namely how classicists can capture intuitionistic logic, the consequence relation is the one for classicists, and the above result shows that the consequence relation is nicely internalized by the classical conditional $\supset$, as it should be the case.
\end{remark}

\begin{remark}
Note that the deduction theorem with respect to $\to$ fails. Indeed, suppose we have the theorem. Then, by applying the deduction theorem twice, we obtain that $\vdash (A{\supset} B){\to} (A{\to} B)$. However, in view of the three-valued model of Lemma~\ref{lem:Not@def} and the soundness of {\bf S} with respect to it, we obtain that $\not\vdash (A{\supset} B){\to} (A{\to} B)$.
\end{remark}

We conclude this section with corollaries of the deduction theorem that will turn out to be crucial for the completeness theorem.

\begin{lemma}\label{lem.trans}
The following rule is derivable in {\bf S}.
\begin{gather}
\frac{\ A{\supset} B\quad B{\supset} C\ }{A\supset C} \tag{Trans} \label{Trans}
\end{gather}
\end{lemma}


\begin{proposition}\label{prop.disj}
For all $\Gamma \cup \{ A , B , C \} \subseteq \mathsf{Form}_\supset$, if $\Gamma,A\vdash C$ and $\Gamma,B\vdash C$ then $\Gamma,A\lor B\vdash C$.
\end{proposition}


\section{Soundness and completeness}\label{sec:soundnesscompleteness}

For the present system, soundness is straightforward by induction on the length of the proof.

\begin{theorem}[Soundness]
For $\Gamma \cup \{ A \}  \subseteq \mathsf{Form}_\supset$, if $\Gamma\vdash A$ then $\Gamma\models A$.
\end{theorem}

For completeness, we first introduce some notions due to Greg Restall in \cite{Restall1994}.

\medskip

\noindent \begin{minipage}{0.48\textwidth}
(i) $\Sigma\vdash_{\Pi} A$ iff $\Sigma\cup\Pi\vdash A$.

(ii) $\Sigma$ is a $\Pi$-\emph{theory} iff:
\begin{itemize}
\setlength{\parskip}{0cm}
\setlength{\itemsep}{0cm}
\item[](a) if $A,B\in\Sigma$ then $A\land B\in\Sigma$.
\item[](b) if $\vdash_{\Pi}A\to B$ then (if $A\in\Sigma$ then $B\in\Sigma$).
\end{itemize}
(iii) $\Sigma$ is \emph{prime} iff $A\in\Sigma$ or $B\in\Sigma$ if $A\lor B\in\Sigma$.

(iv) $\Sigma\vdash_{\Pi}\Delta$ iff $\Sigma\vdash_{\Pi}D_{1}\lor \cdots\lor D_{n}$ 

\qquad \qquad \qquad \qquad for some $D_{1},\ldots,D_{n}\in\Delta$. 
\end{minipage}
\begin{minipage}{0.5\textwidth}
(v) $\vdash_{\Pi}\Sigma\to\Delta$ iff $\vdash_{\Pi} C_{1}{\land}\cdots{\land} C_{n}{\to} D_{1}{\lor}\cdots{\lor} D_{m}$ 

\qquad for some $C_{1},\ldots,C_{n}\in\Sigma$ and $D_{1},\ldots,D_{m}\in\Delta$.

(vi) $\Sigma$ is $\Pi$-\emph{deductively closed} iff $A\in\Sigma$ if $\Sigma\vdash_{\Pi}A$.

(vii) $\langle\Sigma,\Delta\rangle$ is a $\Pi$-\emph{partition} iff:
\begin{itemize}
\setlength{\parskip}{0cm}
\setlength{\itemsep}{0cm}
      \item[](a) $\Sigma\cup\Delta=\mathsf{Form}_\supset$
      \item[](b) $\nvdash_{\Pi}\Sigma\to\Delta$
  \end{itemize}
(viii) $\Sigma$ is \emph{non-trivial} iff $A\notin \Sigma$ for some formula $A$.
\end{minipage}

\smallskip

\begin{lemma}\label{lem.comp1}
If $\Gamma$ is a non-empty $\Pi$-theory, then $\Pi\subseteq\Gamma$.
\end{lemma}

\begin{proof}
Take $A\in\Pi$. Then, we have $\Pi\vdash A$. Since $\Gamma$ is non-empty,
take any $C\in\Gamma$. Then, by \eqref{Ax1}, we obtain $\Pi\vdash C\to A$. Thus, combining this together with $C\in\Gamma$ and that $\Gamma$ is $\Pi$-theory, we obtain $A\in\Gamma$.
\end{proof}

We refer to \cite{DeOmori2014,DeOmori2016,NikiOmoriAiMLsubmitted,Restall1994} for the details of the next lemmas. 

\begin{lemma}\label{lem.comp2}
If $\langle\Sigma,\Delta\rangle$ is a $\Pi$-partition then $\Sigma$ is a prime $\Pi$-theory.
\end{lemma}

\begin{lemma}\label{lem.comp3}
If $\nvdash_{\Pi}\Sigma\to\Delta$ then there are $\Sigma'\supseteq\Sigma$ and $\Delta'\supseteq\Delta$ such that $\langle\Sigma',\Delta'\rangle$ is a $\Pi$-partition.
\end{lemma}

\begin{corollary}\label{cor.comp1}
Let $\Sigma$ be a non-empty $\Pi$-theory, $\Delta$ be closed under disjunction, and $\Sigma\cap\Delta=\emptyset$. Then there is $\Sigma'\supseteq\Sigma$ such that $\Sigma'\cap\Delta=\emptyset$ and $\Sigma'$ is a prime $\Pi$-theory.
\end{corollary}

Note that we make use of \eqref{Trans} and Proposition~\ref{prop.disj} for the next lemma.

\begin{lemma}\label{lem.comp4}
If $\Sigma\nvdash\Delta$ then there are $\Sigma'\supseteq\Sigma$ and $\Delta'\supseteq\Delta$ such that $\langle\Sigma',\Delta'\rangle$ is a partition, and $\Sigma'$ is deductively closed.
\end{lemma}

\begin{corollary}\label{cor.comp2}
If $\Sigma\nvdash A$ then there is $\Pi\supseteq\Sigma$ such that $A\notin\Pi$, $\Pi$ is a prime $\Pi$-theory and is $\Pi$-deductively closed.
\end{corollary}

We also prepare two more lemmas for the truth lemma.

\begin{lemma}\label{lem.comp5}
If $\Delta$ is a $\Pi$-theory and $A\to B\notin \Delta$, then there is a prime $\Pi$-theory $\Gamma\supseteq \Delta$, such that $A\in\Gamma$ and $B\notin\Gamma$.
\end{lemma}

\begin{lemma}\label{lem:pPdc1}
If $\Sigma$ is prime, $\Pi$-deductively closed and $A\not\in \Sigma$ then $A\supset B\in \Sigma$.
\end{lemma}
\begin{proof}
If $\Sigma$ is $\Pi$-deductively closed, then by \eqref{AxC} we obtain $A\lor (A\supset B)\in \Sigma$. This together with $A\not\in \Sigma$ and the primeness of $\Sigma$ implies $A\supset B\in \Sigma$. 
\end{proof}

We are now ready to establish the completeness.

\begin{theorem}[Completeness]\label{thm.comp}
For $\Gamma \cup \{ A \}  \subseteq \mathsf{Form}_{\supset}$, if $\Gamma\models A$ then $\Gamma\vdash A$.
\end{theorem}

\begin{proof}
We prove the contrapositive. Suppose that $\Gamma\nvdash A$. Then, by Corollary \ref{cor.comp2}, there is a $\Pi\supseteq\Gamma$ such that $\Pi$ is a prime $\Pi$-theory, $\Pi$-deductively closed and $A\notin\Pi$. Define the interpretation $\mathfrak{A}=\langle X,\Pi,\leq,I\rangle$, where $X=\{\Delta:\Delta\text{ is a non-empty and non-trivial prime $\Pi$-theory}\}$, $\Delta\leq\Sigma$ iff $\Delta\subseteq\Sigma$ and $I$ is defined thus: for every state $\Sigma$ and propositional parameter $p$:
$I(\Sigma,p)=1\text{ iff }p\in\Sigma$.\\
\indent We show by induction on $B$ that $I(\Sigma,B)=1$ iff $B\in\Sigma$. We will only deal with the cases involving $\to$ and $\supset$ (the cases for conjunction and disjunction can be found, for example, in \cite[Theorem 3.13]{DeOmori2014}). 

\vspace{-1mm}

\begin{itemize}
\item When $B\equiv C\to D$, by IH $I(\Sigma,C\to D)=1$ iff for all $\Delta$ s.t. $\Sigma\subseteq\Delta$, if $C\in\Delta$ then $D\in\Delta$. Hence it suffices to show that this latter condition is equivalent to $C\to D\in\Sigma$. For the forward direction, we argue by contraposition; so assume $C\to D\notin\Sigma$. Then by Lemma \ref{lem.comp5} we can find a non-trivial prime $\Pi$-theory $\Sigma'\supset\Sigma$ such that $C\in\Sigma'$ but $D\notin\Sigma'$. For the backward direction, assume $C\to D\in\Sigma$ and $C\in\Delta$ for any $\Delta$ s.t. $\Sigma\subseteq\Delta$. Then $C\to D\in \Delta$ as well, and so $D\in\Delta$ since $\Delta$ is a $\Pi$-theory.

\item When $B\equiv C\supset D$, by IH $I(\Sigma, C\supset D)=1$ iff $C\not\in \Pi$ or $D\in \Sigma$ iff $C\supset D\in\Sigma$. For the last equivalence, suppose $C \not\in \Pi $. Then, by Lemma~\ref{lem:pPdc1}, we obtain $C\supset D \in \Pi$. This together with Lemma~\ref{lem.comp1} implies that $C\supset D \in \Sigma$. That $D\in \Sigma$ implies $C\supset D\in\Sigma$ is easy in view of \eqref{Kmix}. For the other way around, suppose $C\supset D\in\Sigma$ and $C \in \Pi $; we want to show $D\in\Sigma$.
Then we infer from \eqref{MP3} and \eqref{MP} that $\vdash_{\Pi}(C\supset D)\to D$. Now because $\Sigma$ is a $\Pi$-theory, $C\supset D\in\Sigma$ implies $D\in\Sigma$, as desired.
\end{itemize}
\vspace{-1mm}
It now suffices to observe that $B{\in}\Pi$ for all $B{\in}\Gamma$ and $A{\notin}\Pi$, which in view of the above means $\Gamma\not\models A$. 
\end{proof}

\begin{remark}
It is straightforward to observe that {\bf S} is a conservative extension of intuitionistic logic. In addition, an interpretation at $g$ of formulas in $\{\land,\lor,\supset\}$-fragment corresponds to a classical model, because it never refers to other worlds. Consequently, {\bf S} is also a conservative extension of classical logic. Hence we see that {\bf S} properly combines the two implications.
\end{remark}

\begin{remark}
Given the notion of empirical negation, our classical conditional is defined as the material conditional. However, in view of the treatment of classical conditional by Humberstone in \cite[p.174]{Humberstone1979}, we may also introduce another classical conditional as $\Not (A\land \Not B)$. Of course, these two ways are equivalent in classical logic, but this is not the case in our present context, giving rise to the following truth condition. 

\smallskip

$\bullet$ $I(w, A{\supset} B)=1$ iff $I(g, A)\neq 1$ or $I(g, B)=1$.

\smallskip

\noindent Otherwise, we retain the definition of {\bf S}-model. Let us refer to the resulting model and the semantic consequence relation as {\bf T}-model and $\models_{\bf T}$, respectively. We can then also axiomatize this system.
\end{remark}

\begin{definition}
The system {\bf T} is obtained from {\bf S} by (i) replacing \eqref{AxM3} and \eqref{AxM4} by \eqref{Ksup} and \eqref{MP3}, and (ii) adding the following axiom schemata in {\bf T}:

\vspace{-4mm}

\noindent
\begin{minipage}{.5\textwidth}
\begin{gather}
(A {\supset} B){\to}((\top{\supset} A){\to} B) \tag{T1} \label{T1} 
\end{gather}
\end{minipage}
\begin{minipage}{.5\textwidth}
\begin{gather}
(A{\supset} B){\to}(C{\supset}(A{\supset} B)) \tag{T2} \label{T2}
\end{gather}
\end{minipage}

\smallskip

\noindent where $\top$ abbreviates $p\to p$ for some $p\in \mathsf{Prop}$. We refer to the proof-theoretic consequence relation as $\vdash_{\bf T}$. 
\end{definition}

Then, we may establish that for all $\Gamma \cup \{ A \} \subseteq \mathsf{Form}_{\supset}$, $\Gamma\models_{\bf T} A$ iff $\Gamma\vdash_{\bf T} A$, and that {\bf S} and {\bf T} are incomparable. The details are kept for the full version of this paper, due to space restriction.

\section{Indispensability of disjunction}\label{sec:disjunction}
The axiom schemata of {\bf S} expressing properties of $\supset$ are in implicational form (i.e. $\to$ and $\supset$ are the only connectives), except for \eqref{AxM6}. This allows us to establish the completeness of the conjunction-free fragment of {\bf S}, by altering the definition of $\vdash_{\Pi}\Sigma\to\Delta$ to  ``$\vdash_{\Pi} C_{1}\to(C_{2}\to(\cdots\to(C_{n}\to D_{1}\lor\cdots\lor D_{m})))$ for some $C_{1},\ldots,C_{n}\in\Sigma$ and $D_{1},\ldots,D_{m}\in\Delta$'', and dropping the condition (a) from a $\Pi$-theory. Then we can observe that the arguments in \cite{DeOmori2014,DeOmori2016,NikiOmoriAiMLsubmitted} are suitably modified for conjunction-free language.

It is then a natural question to ask if we may take the fragment of {\bf S} which is also disjunction-free, i.e. the implicational fragment. This will be reduced to the question if we can replace \eqref{AxM6} by an axiom in an implicational form.

It turns out that such axiom does \emph{not} exist. In order to observe this, we change our setting slightly  so as to make the presence of disjunction in axiomatization explicit. We shall employ a formalism with axioms and the rule of substitution, rather than axiom schemata. In addition, we discuss in a system with the absurdity $\bot$. We shall establish a stronger result that even in this expanded language we cannot replace \eqref{AxM6} by an implicational formula.

\begin{definition}
Let {\bf S$_\bot$} be a system in $\mathcal{L}_{\bot, \supset}$ obtained by adding the next axiom schemata to {\bf S}.
\vspace{-2mm}
\begin{gather}
\bot\to A \tag{Ax0} \label{Ax0}    
\end{gather}

\vspace{-2mm}

\noindent If we remove \eqref{AxM6}, then we obtain a system which we shall call {\bf S$^-_\bot$}. If we further eliminate \eqref{Ax0} and move the language to $\mathcal{L}_\supset$, we obtain the system {\bf S$^-$}.
\end{definition}

\begin{definition}
We define {\bf S$_\bot$2} to be a system in $\mathcal{L}_{\bot, \supset}$ defined by the axioms corresponding to axiom schemata in {\bf S$_\bot$}, e.g.
\vspace{-2mm}
\begin{gather}
    (p{\supset}r){\to}((q{\supset}r){\to}((p\lor q){\supset}r)) \label{AxM6'} \tag{AxM6'}
\end{gather}

\vspace{-2mm}

\noindent with rules \eqref{MP} and 

\vspace{-2mm}
\begin{gather}
\frac{A}{A[p/B]} \label{Sub} \tag{Sub} 
\end{gather}

\vspace{-2mm}

\noindent where $A[p/B]$ denotes the result of substituting all instances of $p$ occurring in $A$ (and assumption) by $B$. If we remove \eqref{AxM6'} then it defines a subsystem {\bf S$^-_\bot$2}.
\end{definition}

\begin{remark}
It is straightforward to observe {\bf S$_\bot$} and {\bf S$_\bot$2}, as well as {\bf S$^-_\bot$} and {\bf S$^-_\bot$2} coincide. Furthermore, Propositions~\ref{prop.inf}, \ref{prop:DT} as well as Lemma \ref{lem.trans} do not depend on \eqref{AxM6}, and so they hold for {\bf S$^-_\bot$} as well. 
The reason we introduced $\bot$ in the language is to appeal to the intuitionistic modal logic {\bf L$_4$} introduced in \cite{Ono1977} by Hiroakira Ono. Let us briefly recall this system.
\end{remark}

\begin{definition}[Ono]
The logic {\bf L$_4$} is defined in 
$\mathcal{L}_{\bot, \Box}$ with the following axiom schemata and rules, along with those of intuitionistic logic.

\vspace{-3mm}
\noindent
\begin{minipage}{.5\textwidth}
\begin{gather}
\Box (A{\to} B){\to} (\Box A{\to} \Box B) \label{@1} \tag{$\Box$1} \\
\Box A{\to} A \label{@2} \tag{$\Box$2} \\
\Box A{\to} \Box \Box A \label{@3} \tag{$\Box$3} \\
\Box A\lor \Box (\Box A{\to} B) \label{@4} \tag{$\Box$4}
\end{gather}
\end{minipage}
\begin{minipage}{.5\textwidth}
\begin{gather}
\frac{A}{ \Box A } \qquad \quad \label{RN} \tag{RN} \\
\frac{\ A\quad A{\to} B \ }{B} \tag{MP2}
\end{gather}
\end{minipage}
\end{definition}

\begin{remark}
{\bf L$_{4}$} is a syntactic variant of the logic {\bf TCC$_\omega$}, introduced in \cite{Gordienko}, which in turn is a subsystem of the expansion of intuitionistic logic by empirical negation, {\bf IPC$^\Not$} introduced in \cite{MikeEN}. See also \cite{NikiBSL2020I,NikiBSL2020II,NikiOmoriAiMLsubmitted} for a comparison of these systems.
\end{remark}

We are now going to establish the relationship between {\bf S$^-_\bot$} and {\bf L$_{4}$}.

\begin{definition}
Let $()^{\Box }$ and $()^{\supset}$ be the following translations between $\mathcal{L}_{\bot,\supset}$ and $\mathcal{L}_{\bot, \Box}$.
\vspace{-3mm}
\begin{IEEEeqnarray*}{rClrCl}
p^{\Box }            & = & p                     &\hspace{3mm} p^{\supset}         & = & p.\\
\bot^{\Box }         & = & \bot                  &\hspace{3mm} \bot^{\supset}      & = & \bot.\\
(A\circ B)^{\Box }   & = & A^{\Box } \circ B^{\Box }     &\hspace{3mm}(A\circ B)^{\supset} & = & A^{\supset} \circ B^{\supset}.\\
(A\supset B)^{\Box } & = & \Box A^{\Box }\to B^{\Box }       &\hspace{3mm} (\Box A)^{\supset} & = & (A^{\supset}\supset\bot)\supset\bot.
\end{IEEEeqnarray*}

\noindent where $\circ\in\{\land,\lor,\to\}$.
\end{definition}

We shall occasionally abbreviate $A\supset\bot$ as $\Not A$. The choice of the symbol $\Not$ is based on the fact that the formula defines empirical negation in {\bf $S_{\bot}$}.

\begin{lemma}\label{lem.ex}
The following formulas are provable in {\bf S}.

\vspace{-3mm}
\noindent 
\begin{minipage}{.5\textwidth}
\begin{gather}
(A{\supset}(B{\supset}C))\to(B{\supset}(A{\supset}C)) \tag{Ex} \label{Ex}
\end{gather}
\end{minipage}
\begin{minipage}{.5\textwidth}
\begin{gather}
(B\to C)\to(A{\supset}B\to A{\supset}C) \tag{Pfix} \label{Pfix}
\end{gather}
\end{minipage}
\end{lemma}

\begin{proof}
For \eqref{Ex}, on one hand $B{\supset}(A{\supset} B)$ from \eqref{Ksup}. On the other hand, $(A{\supset} B){\supset}((A{\supset}(B{\supset} C){\to}(A{\supset} C)))$ from \eqref{AxM2} and \eqref{AxM4}. Hence by \eqref{Trans}  $B{\supset}((A{\supset}(B{\supset} C){\to}(A{\supset} C)))$; then apply \eqref{AxM4} again to obtain the formula. For \eqref{Pfix}, By \eqref{MP3} and deduction theorem, $A\vdash (A{\supset} B){\to} B$. Thus $A\vdash (B\to C)\to ((A\supset B){\to} C)$. Therefore by deduction theorem and \eqref{AxM3}, we conclude $(B{\to} C){\to}(A{\supset}B{\to} A{\supset}C)$. 
\end{proof}

\begin{lemma}\label{lem.trans.pr}
\textnormal{(i)} ${\bf S^{-}_{\bot}}\vdash A\leftrightarrow (A^{\Box })^{\supset}$;
\textnormal{(ii)} ${\bf L_{4}}\vdash A\leftrightarrow(A^{\supset})^{\Box }$.
\end{lemma}
\begin{proof}
By induction on $A$.
For (i), we need ${\bf S^{-}_{\bot}}\vdash (B{\supset} C) {\leftrightarrow} ((B{\supset} C)^\Box )^{\supset}$. For this it suffices to show that:

\smallskip

$\bullet$ ${\bf S^{-}_{\bot}}\vdash (B\supset C)\leftrightarrow (\Not\Not B\to C)$.

\smallskip

\noindent For the left-to-right direction, because $B\supset((B\supset C)\to C)$ and $\Not B\supset(\Not\Not B\to C)$  are provable from instances of \eqref{MP3}, we obtain $B\supset((B\supset C)\to (\Not\Not B\to C))$ and $\Not B\supset((B\supset C)\to(\Not\Not B\to C))$ by \eqref{Pfix}. Then the desired implication follows from \eqref{AxM5}. For the right-to-left direction, first we note $B\supset\Not\Not B$ follows from $(B\supset\bot)\supset(B\supset\bot)$ by \eqref{Ex}. Also it is a consequence of \eqref{Pfix} that $(B\supset\Not\Not B)\to(B\supset((\Not\Not B\to C)\to C))$. Thus by \eqref{MP2} and \eqref{AxM3} we obtain the desired implication.

For (ii), we need ${\bf L_4}\vdash \Box A\leftrightarrow((\Box A)^{\supset})^{\Box }$. It suffices to show 
${\bf L_4}\vdash \Box A\leftrightarrow\neg\Box \neg\Box  A$.
This follows from \eqref{Ax0} and \eqref{@4}. 
\end{proof}

\begin{proposition}\label{prop.trans.pr}
\textnormal{(i)} ${\bf S^{-}_{\bot}}\vdash A$ iff ${\bf L_{4}}\vdash A^\Box $;
\textnormal{(ii)} ${\bf L_{4}}\vdash A$ iff ${\bf S^{-}_{\bot}}\vdash A^\supset$.
\end{proposition}

\begin{proof}
We show by induction on $A$. By Lemma \ref{lem.trans.pr} it is sufficient to show the left-to-right directions. 

For (i), we need to consider the cases for \eqref{AxM1}--\eqref{AxM5} and \eqref{MP}. The translation of \eqref{AxM5} follows from \eqref{@4} and that of \eqref{AxM1} from \eqref{@1} and \eqref{@2}.
The translations of \eqref{AxM3} and \eqref{AxM4} follow by intuitionitistic logic, and that of \eqref{MP} from \eqref{RN} and \eqref{MP2}. For \eqref{AxM2}, we require
\begin{itemize}
\item ${\bf L_{4}}\vdash(\Box A\to(\Box B\to C))\to(\Box (\Box A\to B)\to(\Box A\to C))$.    
\end{itemize}
for which is enough to show
${\bf L_{4}}\vdash(\Box (\Box A\to B)\land \Box A)\to (\Box A\land \Box B)$.    
This follows from \eqref{@1} and \eqref{@3}.\\
\indent For (ii), we have to consider the cases for \eqref{@1}--\eqref{@4} and \eqref{RN}. We recall the equivalence observed in Lemma \ref{lem.trans.pr}, namely $(A\supset B){\leftrightarrow} (\Not\Not A\to B)$. Then the translation of \eqref{@2} follows form $A\supset A$. That of \eqref{@3} follows from $\Not\Not A{\to}\Not\Not\Not\Not A$, which follows from \eqref{AxC} and \eqref{MP3}. For \eqref{RN}, if ${\bf S^-_\bot}\vdash A$ then ${\bf S^-_\bot}\vdash A\supset\Not\Not A$ implies ${\bf S^-_\bot}\vdash \Not\Not A$.
The translation of \eqref{@4} is derived from \eqref{AxM5}; we have $\Not\Not A\supset\Not\Not A$ and $\Not A\supset(\Not\Not(\Not\Not A\to B))$. The latter follows from $\Not\Not A\to(\Not A\supset B)$, which by \eqref{AxM4} derives $\Not A\supset(\Not\Not A\to B)$. Then use deduction theorem and the argument for \eqref{RN}. For \eqref{@1}, we need to show
\begin{itemize}
\item ${\bf S^-_\bot}\vdash\Not\Not(A\to B)\to(\Not\Not A\to\Not\Not B)$
\end{itemize}
This is equivalent to
    ${\bf S^-_\bot}\vdash(A\to B)\supset(A\supset \Not\Not B)$
Now as $A, A\to B\vdash B$ holds in {\bf S$^-_\bot$}. Then arguing as in the case for \eqref{RN}, we obtain $A, A\to B\vdash \Not\Not B$, from which the above formula is derived. 
\end{proof}

The following semantics corresponds to {\bf L$_4$}, as we observed in \cite{NikiOmoriAiMLsubmitted}.

\begin{definition}
An {\bf L$_4$}-model for $\mathcal{L}_{\bot, \Box}$ is a triple $\langle W, \leq, V \rangle$ as in {\bf S}-model, except that we drop the base state.
The interpretations $I_{\Box }$ to state-formula pairs are the same with {\bf S}-models except for the following:

\smallskip

$\bullet$ $I_{\Box }(w,\bot)=0$;

$\bullet$ $I_{\Box }(w, \Box A)=1$ iff for all $x\in W$: $I_{\Box }(x, A)= 1$.

\smallskip

\noindent Finally, the semantic consequence is now defined as follows: $\Sigma \models_{\Box } A$ iff for all {\bf L$_4$}-models $\langle W, \leq, V \rangle$, if $I_{\Box }(w, B)=1$ for all $B\in \Sigma$ and for all $w\in W$, then $I_{\Box }(w, A)=1$ for all $w\in W$.
\end{definition}

\begin{theorem}[\cite{NikiOmoriAiMLsubmitted}]\label{thm.comp@}
${\bf L_{4}}\vdash A$ iff $\models_{\Box } A$.
\end{theorem}

An analogous model can be considered for the language $\mathcal{L}_{\bot,\supset}$, with the only difference being the interpretation (which we shall denote by $I_{\supset}$) where

\smallskip

$\bullet$ $I_{\supset}(w, A{\supset} B)=1$ iff $I_{\supset}(x, A)\neq 1$ for some $x\in W$ or $I_{\supset}(w, B)=1$.

\smallskip

\noindent We shall call a model of this semantics an {\bf S$^{-}_{\bot}$}-model,
and denote the semantic consequence for this semantics by $\models_{\supset}$. Note that the class of {\bf S$^{-}_{\bot}$}-models with a base state is nothing but the class of {\bf S}-models (with $\bot$ added to the language), because the interpretations of $\supset$ coincide. We shall use $\models_{\bot}$ to denote the semantic consequence restricted to this class of models (to be called {\bf S$_\bot$}-models). Then the next completeness follows analogously to Theorem \ref{thm.comp}.

\begin{proposition}
${\bf S_\bot}\vdash A$ iff $\models_{\bot}A$.
\end{proposition}

\begin{proposition}\label{prop.trans.sem}
\textnormal{(i)} $\models_{\supset} A$ iff $\models_{\Box } A^{\Box }$;
\textnormal{(ii)} $\models_{\Box} A$ iff $\models_{\supset} A^{\supset}$.
\end{proposition}

\begin{proof}
We note each {\bf L$_4$}-model can be seen as a {\bf S$^-_\bot$}-model, and vice versa. Then it suffices to show by induction on $A$ that $I_{\supset}(w,A)=1$ iff $I_{\Box }(w,A^{\Box })=1$, and similarly for the other translation.\\
\indent For instance, for $A\equiv B\supset C$, 
\vspace{-3mm}
\begin{IEEEeqnarray*}{rCl}
I_{\supset}(w,B\supset C)=1 & \text{ iff } & \forall{w'}(I_{\supset}(w',B)=1)\text{ implies } I_{\supset}(w,C)=1.\\
                            & \text{ iff } & \forall{w'}(I_{\Box }(w',B^{\Box })=1)\text{ implies } I_{\Box }(w,C^{\Box })=1. \text{ (by IH)}\\
                            & \text{ iff } & \forall{w'\geq w}(I_{\Box }(w',\Box B^{\Box })=1\text{ implies } I_{\Box }(w',C^{\Box })=1).\\
                            & \text{ iff } & I_{\Box }(w,(B\supset C)^{\Box })=1.
\end{IEEEeqnarray*}
We safely leave the other cases to the reader.
\end{proof}

As a consequence of Theorem \ref{thm.comp@}, Propositions \ref{prop.trans.pr} and \ref{prop.trans.sem}, we obtain the following.

\begin{corollary}\label{cor.comp.minus}
${\bf S^-_\bot}\vdash A$ iff $\models_{\supset} A$.
\end{corollary}


\begin{definition}\label{def.consv}
Let $\mathcal{M}=\langle W,\leq, V\rangle$ be an {\bf S$^-_\bot$}-model. 
We define another {\bf S$^-_\bot$}-model $\mathcal{M}'=\langle W',\leq',V'\rangle$ by setting $W'=W\cup\{g\}$; $\leq'=\leq\cup\{(g,w):w\in W'\}$ and: $V'(w,p)=1 \text{ iff } V(w,p)=1\text{ or } (w=g\text{ and } V(w',p)=1\text{ for all }w'\in W)$.
\end{definition}

Note that $V'$ is well-defined as a valuation.
We use the notation $\mathcal{M}\models_{\supset} A$ to abbreviate $I_{\supset}(w,A)=1$ \emph{for all} $w\in W$ \emph{in} $\mathcal{M}$.

\begin{lemma}\label{lem.consv}
Let $\mathcal{M}$ and $\mathcal{M}'$ be as in the previous definition. Then for any formula $A$ in $\mathcal{L}^-_{\bot, \supset}$ (i.e. disjunction-free):
\textnormal{(i)} $I_{\supset}(w,A)=1$ iff $I'_{\supset}(w,A)=1$ for any $w\in W$, and 
\textnormal{(ii)} $\mathcal{M}\models_{\supset}A$ iff $I'_{\supset}(g,A)=1$.
\end{lemma}

\begin{proof}
We show (i) and (ii) by simultaneous induction on $A$. Here we treat the cases for $\to$ and $\supset$.
If $A\equiv A_{1}\to A_{2}$ then for (i):
\vspace{-3mm}
\begin{IEEEeqnarray*}{rCl}
I_{\supset}(w,A_{1}\to A_{2})=1 & \text{ iff } & \forall w'\geq w(I_{\supset}(w',A_{1})=1\text{ implies }I_{\supset}(w',A_{2})=1 ).\\
                                  & \text{ iff } & \forall w'\geq w(I'_{\supset}(w',A_{1})=1\text{ implies }I'_{\supset}(w',A_{2})=1 ).\\
                                  & \text{ iff } & I'_{\supset}(w,A_{1}\to A_{2})=1.
\end{IEEEeqnarray*}
and for (ii), first we note
\vspace{-3mm}
\begin{IEEEeqnarray*}{rCl}
\mathcal{M}\models_{\supset}A_{1}\to A_{2} & \text{ iff } & \forall w\in W(I_{\supset}(w,A_{1})=1\text{ implies }I_{\supset}(w,A_{2})=1 ).\\
                                            & \text{ iff } & \forall w\in W(I'_{\supset}(w,A_{1})=1\text{ implies }I'_{\supset}(w,A_{2})=1 ).
\end{IEEEeqnarray*}
Also if $I'_{\supset}(g,A_{1})=1$, then by IH $\mathcal{M}\models_{\supset}A_{1}$. So $\mathcal{M}\models_{\supset}A_{2}$, and by IH again $I'_{\supset}(g,A_{2})=1$. Combining the above two observations, we conclude $I'_{\supset}(g,A_{1}\to A_{2})=1$. For the other direction, if $I'_{\supset}(g,A_{1}\to A_{2})=1$ then $\forall w\in W(I'_{\supset}(w,A_{1})=1$ implies $I'_{\supset}(w,A_{2})=1 )$. Thus by the equivalence above, $\mathcal{M}\models_{\supset}A_{1}\to A_{2}$.

\indent If $A\equiv A_{1}\supset A_{2}$, then for (i):
\vspace{-3mm}
\begin{IEEEeqnarray*}{rCl}
I_{\supset}(w,A_{1}\supset A_{2})=1 & \text{ iff } & \mathcal{M}\models_{\supset}A_{1}\text{ implies }I_{\supset}(w,A_{2})=1.\\
                                  & \text{ iff } & I'_{\supset}(g,A_{1})=1\text{ implies }I'_{\supset}(w,A_{2})=1.\\
                                  & \text{ iff } & I'_{\supset}(w,A_{1}\supset A_{2})=1.
\end{IEEEeqnarray*}
and for (ii):
\vspace{-4mm}
\begin{IEEEeqnarray*}{rCl}
\mathcal{M}\models_{\supset}A_{1}\supset A_{2} & \text{ iff } & \mathcal{M}\models_{\supset}A_{1}\text{ implies }\mathcal{M}\models_{\supset}A_{2}.\\
                                  & \text{ iff } & I'_{\supset}(g,A_{1})=1\text{ implies }I'_{\supset}(g,A_{2})=1.\\
                                  & \text{ iff } & I'_{\supset}(g,A_{1}\supset A_{2})=1.
\end{IEEEeqnarray*}
This completes the proof.
\end{proof}

\begin{theorem}\label{thm.consv}
Let $A$ be a formula in $\mathcal{L}^-_{\bot, \supset}$. Then $\models_{\bot} A$ implies $\models_{\supset}A$. 
\end{theorem}

\begin{proof}
If $\not\models_{\supset} A$, there is an {\bf S$^-_\bot$}-model $\mathcal{M}$ with $\mathcal{M}\not\models_{\supset} A$. By Lemma~\ref{lem.consv}, there is another {\bf S$^-_\bot$}-model $\mathcal{M}'$ with a base state such that $\mathcal{M}'\not\models_{\supset} A$. But as we observed $\mathcal{M}'$ is nothing but an {\bf S$_\bot$}-model; so $\not\models_{\bot} A$. 
\end{proof}

\begin{corollary}\label{cor.consv}
Let $A$ be a formula in $\mathcal{L}^-_{\bot, \supset}$. Then ${\bf S_\bot2}\vdash A$ implies
 ${\bf S^-_\bot2}\vdash A$.
\end{corollary}

\begin{proof}
If ${\bf S_{\bot}}2\vdash A$, then ${\bf S_{\bot}}\vdash A$ and so $\models_{\bot} A$ by soundness. So by Theorem \ref{thm.consv}, $\models_{\supset}A$. Hence by Corollary \ref{cor.comp.minus}, ${\bf S^{-}_{\bot}}\vdash A$ and thus ${\bf S^{-}_{\bot}2}\vdash A$. 
\end{proof}

\begin{corollary}\label{cor.disj}
There is no $A{\in} \mathsf{Form}^-_{\bot, \supset}$ such that (i) ${\bf S_{\bot}2}\vdash A$ and (ii) the system ${\bf S^{-}_{\bot}2}+A$ derives \eqref{AxM6'}.
\end{corollary}

\begin{proof}
If ${\bf S_{\bot}2}\vdash A$ then by Corollary \ref{cor.consv}, ${\bf S^{-}_{\bot}2}\vdash A$ already. Thus ${\bf S^{-}_{\bot}2}+A={\bf S^{-}_{\bot}2}$.
However, we can show $\not\models_{\supset}\eqref{AxM6'}$ by considering a model $(W,\leq,V)$ where $W{=}\{w,w'\}$, $\leq=\{(w,w),(w',w')\}$, $V(w,p)=V(w',q)=1$ and $V(w,q){=}V(w',p){=}V(w,r){=}V(w',r){=}0$. Then $I(w,p{\supset} r){=}1$ and $I(w,q{\supset} r){=}1$, but $I(w,(p{\lor} q){\supset} r){=}0$. Thus by soundness, ${\bf S^{-}_{\bot}}\nvdash \eqref{AxM6'}$ and so ${\bf S^{-}_{\bot}2}\nvdash \eqref{AxM6'}$. Therefore the second condition cannot be satisfied. 
\end{proof}

\begin{remark}
If there is a formula which satisfies the above conditions with respect to {\bf S$^-$2} (i.e. the system in $\mathcal{L}_{\supset}$), then it is derivable in {\bf S$_\bot$2} and also derives \eqref{AxM6'} when added to {\bf S$^-_\bot$2}. This is impossible by the corollary; so \eqref{AxM6'} cannot be converted to a disjunction-free axiom for the language $\mathcal{L}_{\supset}$ either.
\end{remark}

We conclude this section by observing that a disjunction-free rule cannot replace \eqref{AxM6'} in ${\bf S_{\bot}2}$. Recall that a rule \AxiomC{$A_{1},\cdots,A_{n}$}\UnaryInfC{$B$}\DisplayProof is said to be \emph{derivable} in a proof system if there is a proof in the system of $B$ from $A_{1},\cdots,A_{n}$. Moreover, a rule \AxiomC{$A_{1},\cdots,A_{n}$}\UnaryInfC{$B$}\DisplayProof is \emph{admissible} if the provability of $A_{1},\cdots,A_{n}$ implies the provability of $B$ (see \cite{TroelstraSchwichtenberg2000} for more details). Then, we obtain the following result.

\begin{corollary}\label{cor.disj2}
There is no rule $\mathsf{R}$ of the form \AxiomC{$A_{1},\cdots,A_{n}$}\UnaryInfC{$B$}\DisplayProof, where $B$ is in $\mathcal{L}^-_{\bot, \supset}$, such that (i) $\mathsf{R}$ is admissible in ${\bf S_{\bot}2}$; and (ii) ${\bf S^{-}_{\bot}2}+\mathsf{R}$ derives \eqref{AxM6'}.
\end{corollary}

\begin{proof}
It suffices to show that if (i) is satisfied, then $\mathsf{R}$ is already admissible in ${\bf S^{-}_{\bot}2}$. Now if ${\bf S^{-}_{\bot}2}\vdash A_{1},\ldots{\bf S^{-}_{\bot}2}\vdash A_{n}$, then ${\bf S_{\bot}2}\vdash A_{1},\ldots{\bf S_{\bot}2}\vdash A_{n}$. Hence ${\bf S_{\bot}2}\vdash B$ by $\mathsf{R}$, but then by Corollary \ref{cor.consv}, ${\bf S^{-}_{\bot}2}\vdash B$. Hence $\mathsf{R}$ is already admissible in ${\bf S^{-}_{\bot}2}$. 
\end{proof}

\section{Adding a weaker absurdity}\label{sec:weak-absurdity}
In order to introduce negation into our language, recall that the intuitionists treat negation as implication to absurdity $\neg A:=A{\to} \bot$. Moreover, intuitionists assume that $\bot$ satisfies the law of explosion $\bot{\to} A$; if no such assumption is made, we obtain Johansson's minimal logic, introduced in \cite{Johansson1937}.

Now, in order to obtain classical negation at the base state, we would need absurdity to be explosive therein. However, this does not necessarily require absurdity to be explosive elsewhere. In other words, the law of explosion may be allowed in a restricted manner, much like the law of excluded middle. This allows us, somewhat surprisingly, to obtain a system, which we refer to as {\bf S$_{\bot_{w}}$}, that combines classical logic and minimal logic. 

\begin{definition}
Let {\bf S$_{\bot_{w}}$} be an expansion of {\bf S} in $\mathcal{L}^-_{\bot, \supset}$ with the following additional axiom.
\vspace{-2mm}
\begin{gather}
\bot{\supset}A \tag{AxE} \label{AxE}    
\end{gather}

\vspace{-2mm}

\noindent We refer to the derivability in {\bf S$_{\bot_{w}}$} as $\vdash_{\bot_{w}}$.
\end{definition}

\begin{remark}
Semantically, we can basically reuse the Kripke model for {\bf S}. The only change we need to make is to expand $V$ to be $V: W\times \mathsf{Prop}\cup\{\bot\} \longrightarrow \{ 0 , 1 \}$ with a condition that $V(g,\bot)=0$ for all $V$. Without this condition and the requirement that there is a base state, the $\supset$-free fragment of this semantics is equivalent to that of minimal logic in \cite{segerberg1968propositional}.
Then $I(w,\bot)$ is defined to be equal to $V(w,\bot)$. We refer to the semantic consequence relation as $\models_{\bot_{w}}$. 
\end{remark}

\begin{theorem}[Soundness and Completeness]\label{thm.mincomp}
For $\Gamma \cup \{ A \}  \subseteq \mathsf{Form}_{\bot,\supset}$, $\Gamma\vdash_{\bot_{w}} A$ iff $\Gamma\models_{\bot_{w}} A$.
\end{theorem}

\begin{proof}
For soundness, it suffices to check \eqref{AxE} holds in any model. This is immediate from the fact that $I(g,\bot)=0$ in any model. For completeness, we can use the same countermodel construction as in Theorem~\ref{thm.comp}. That is to say, we set $I(\Sigma,\bot)=1$ iff $\bot\in\Sigma$. Then the base state $\Pi$ is deductively closed, so If $\bot\in\Pi$ then by \eqref{AxE} and \eqref{MP} we deduce $A\in\Pi$ for all $A$, a contradiction. Thus $\bot\notin\Pi$ and so the condition for our model that $V(\Pi,\bot)=0$ is satisfied. 
\end{proof}

Next we check that the intuitionistic absurdity (which we denote by $\bot_{i}$ here) satisfying the law of explosion cannot be defined in {\bf S$_{\bot_{w}}$}. 

\begin{proposition}\label{prop.bot}
There is no formula $F\in\mathsf{Form}_{\bot,\supset}$ such that $I(w,F)=0$  for any $w\in W$ in any {\bf S$_{\bot_{w}}$}-model $\langle g,W,\leq,V\rangle$.
\end{proposition}

\begin{proof}
Suppose there is such a formula. Choose a model such that there are more than two states, and (i) $V(w,\bot)=1$ iff $w\in W\backslash\{g\}$; (ii) $V(w,p)=1$ iff $w\in W\backslash\{g\}$ for all $p$ occurring in $F$; (iii) $V(w,q)=1$ iff $w\in\emptyset$ for other propositional variables. Then we can show by induction that for any subformula $A$ of $F$, $I(w,A)=1$ iff $w\in W$, or $I(w,A)=1$ iff $w\in W\backslash\{g\}$. Then $I(w, F)=1$, a contradiction. 
\end{proof}

Finally, we show that {\bf S$_{\bot_{w}}$} is a conservative extension of minimal logic. Let \textbf{MPC} be the propositional minimal logic defined by \eqref{Ax1}--\eqref{Ax8} and \eqref{MP} in the language $\mathcal{L}_{\bot}=\{\land,\lor,\to,\bot\}$. As we mentioned above, the Kripke semantics for \textbf{MPC} is obtainable from that of {\bf S$_{\bot_{w}}$} by posing no restriction on the valuation of $\bot$ and the shape of a model. By making use of $\vdash_{j}$ and $\models_{j}$ to denote the derivability and validity in \textbf{MPC}, respectively, we obtain the following result. 

\begin{theorem}\label{thm.bot}
For $A\in \mathsf{Form}_{\bot}$, if $\vdash_{\bot_{w}} A$, then $\vdash_{j} A$.
\end{theorem}

\begin{proof}
By Theorem \ref{thm.mincomp}, if $\vdash_{\bot_{w}} A$ then $\models_{\bot_{w}} A$. Let $\mathcal{M}=\langle W,\leq,V\rangle$ be an \textbf{MPC}-model. Take $w\in W$. We wish to show $I(w, A)=1$. Towards this, we first note that similarly to intuitionistic logic, we can take the \emph{truncated} model \cite[p.78]{TroelstravanDalen1988} $\mathcal{M}'$ of $\mathcal{M}$, i.e. the submodel restricted to the worlds above $w$. In particular $I(w,A)=1$ iff $I'(w,A)=1$. Moreover, $\mathcal{M}'$ is also a truncated model of the model
$\mathcal{M}_{0}$ obtained by adding a new root $w_{0}$ to $\mathcal{M}'$, with $V_{0}(w_{0},B)=0$ for $B\in\{p,\bot\}$. Then $\mathcal{M}_{0}$ is well-defined, and can be seen as an {\bf S$_{\bot_{w}}$}-model because it satisfies the condition that $\bot$ is not true at the root. Hence by assumption, $I_{0}(w_{0},A)=1$ and so $I_{0}(w,A)=1$. Therefore  $I'(w,A)=1$ and consequently $I(w,A)=1$, as required. Thus $\models_{j} A$ and by the completeness of \textbf{MPC}, $\vdash_{j} A$. 
\end{proof}

\section{Comparison with \textbf{CIPC}}\label{sec:comparison}

So far, we have seen how we may obtain a combination of classical positive logic and intuitionistic positive logic, which can be expanded to combine (i) classical logic and intuitionistic logic, by adding $\bot{\to} A$, as well as (ii) classical logic and minimal logic, by adding $\bot{\supset} A$. In view of these results, we believe that our system offers a ground for classicists to capture intuitionistic logic. Our system, however, is not the first attempt of combining classical positive logic and intuitionistic positive logic. Indeed, there is an attempt by Carlos Caleiro and Jaime Ramos who presented a system that combines classical conditional and intuitionistic conditional in \cite{Caleiro2007combining}. Since this will give us an opportunity to highlight some of the features of our system, we will focus on their system in this section.

The system {\bf CIPC} of Caleiro and Ramos, introduced in \cite{Caleiro2007combining}, is a logic with classical implication ($\Rightarrow$) and intuitionistic implication ($\to$), obtained via the technique of \emph{cryptofibring} (cf. \cite{caleiro2007fibring}). Semantically, it is captured by intuitionistic Kripke models with a base state together with the following condition for the classical conditional.

$\bullet$ $I(w,A\Rightarrow B)=1\text{ iff for some } w'\leq w\ (I(w',A)=0\text{ or }I(w',B)=1)$.

\noindent As observed in \cite[Corollary 2]{Caleiro2007combining}, this is equivalent to $I(g,A)=0\text{ or }I(w,B)=1$.
Moreover the system has two sorts of atomic formulas, classical and intuitionistic. Classical ones are either true in every world in the model, or true nowhere. Intuitionistic ones are evaluated as usual.\footnote{A similar idea is employed by Hidenori Kurokawa in \cite{Kurokawa2009hypersequent} to formulate intuitionistic logic with classical atoms. 
} Then \textbf{CIPC} has the following axiomatization.

\vspace{-3mm}

\noindent 
\begin{minipage}{.49\textwidth}
\begin{gather}
A\Rightarrow(B\Rightarrow A) \label{C1} \tag{C1}\\
(A{\Rightarrow}(B{\Rightarrow}C)){\Rightarrow}((A{\Rightarrow} B){\Rightarrow}(A{\Rightarrow}C)) \label{C2} \tag{C2}\\
((A\Rightarrow B)\Rightarrow A)\Rightarrow A \label{C3} \tag{C3}\\
A\to(B\to A) \label{I1} \tag{I1}\\
(A{\to}(B{\to}C)){\to}((A{\to}B){\to}(A{\to}C)) \label{I2} \tag{I2}\\
\frac{\ A\quad A{\Rightarrow} B \ }{B} \label{CMP} \tag{CMP}
\end{gather}
\end{minipage}
\begin{minipage}{.49\textwidth}
\begin{gather}
A\to(B\Rightarrow A) \label{X1} \tag{X1}\\
(A^{*}\Rightarrow B)\to(A^{*}\to B) \label{X2} \tag{X2}\\
A\Rightarrow((A\Rightarrow B)\to(A\to B)) \label{X3} \tag{X3}\\
(X{\Rightarrow}(A{\to}B)){\to}((X{\Rightarrow}A){\to}(X{\Rightarrow}B)) \label{X4} \tag{X4}\\
\frac{\ A\quad A{\to} B \ }{B} \label{IMP} \tag{IMP} 
\end{gather}
\end{minipage}

\noindent where $A^{*}$ is \emph{classical}, i.e. it is built up from classical atoms by means of $\Rightarrow$.\\
\indent As the above semantic condition makes it clear, {\bf CIPC} is almost identical to {\bf S} except for the restriction of the connectives and the presence of classical atoms. Now, Definition \ref{def.consv} can be naturally extended to treat classical atoms, and consequently  Corollary \ref{cor.disj} applies to {\bf CIPC} as well. Therefore, for {\bf CIPC} in the full language with conjunction and disjunction, it is insufficient to add standard axiom for $\land$ and $\lor$. Indeed, if we adopt the interpretation of $\supset$ in {\bf S$^-$}, namely

$\bullet$ $I(w, A{\Rightarrow} B)=1\text{ iff for some } x\ (I(x, A)\neq 1)\text{ or }I(w, B)=1$,

\noindent then it is straightforward to observe that the axioms of {\bf CIPC} are sound with respect to the class of Kripke models without the condition that there is a base state. Now if {\bf CIPC} in the full language is complete with respect to the class of Kripke models with a base state, then
$A\Rightarrow C\to((B\Rightarrow C)\to(A\lor B\Rightarrow C))$
has to be provable in {\bf CIPC}, and so it must be valid in all Kripke models without the condition of a base state. However analogously to what we observed in in Corollary \ref{cor.consv} this is impossible. Therefore {\bf CIPC} in the full language cannot be complete with respect to the class of Kripke models with a base state.\\
\indent Another thing to note here is that the above three interpretations for $\Rightarrow$ coincide only in a Kripke model with a base state. If we do not assume a base state, then the first and the third condition are no longer equivalent. In particular, the first condition becomes equivalent to that of $\Not A\lor B$ in Graham Priest's logic of co-negation \cite{PriestdaCL}. This difference illustrates why \eqref{X3} is said to fail when considering arbitrary Kripke models in \cite[Section 5]{Caleiro2007combining}. This occurs when the first condition is considered, but not when the third condition is considered, as we discussed above.

\section{Concluding remarks}\label{sec:conclusion}
In this paper, motivated by questions concerning the relation between classical logic and intuitionistic logic, we introduced an expansion of positive intuitionistic logic by classical conditional. Semantically, we made essential use of the base point in defining the classical conditional, and proof-theoretically, we introduced the system as an axiomatic proof system. We then established soundness and strong completeness results. Moreover, by taking this system as the basic system, we 
discussed the indispensability of disjunction in \S\ref{sec:disjunction}, pointed out that we may also combine classical logic and minimal logic in \S\ref{sec:weak-absurdity}, and quickly compared with one of the most closely related proposals, namely the system {\bf CIPC} due to Caleiro and Ramos in the literature in \S\ref{sec:comparison}. The following table summarizes the systems we looked at in the paper.

\vspace{-2mm}
\begin{table}[H]
\begin{center}
\begin{tabular}{r | l}
System & \\
\hline
{\bf S} & New combination of classical and intuitionistic conditionals\\
{\bf T} & Alternative combination incomparable with {\bf S}\\
${\bf S}_{\bot_{w}}$ & Combination of classical and minimal conditionals\\
\hline
$()_{\bot}$ & Addition of intuitionistic $\bot$ to the system\\
$()^{-}$ & Subsystem without \eqref{AxM6} (or \eqref{AxM6'})\\
$(){\bf 2}$ & System defined with substitution rule\\
\hline
{\bf MPC}    & Johansson's minimal logic\\
${\bf L_{4}}$ & One of intuitionistic {\bf S5} of Ono\\
{\bf CIPC} & Combination of classical and intuitionistic conditionals by Caleiro and Ramos
\end{tabular}
\end{center}
\caption{List of systems discussed in this paper}
\end{table}

\vspace{-3mm}

Further results we are ready to report, though had to be kept aside for the full version due to space restriction, include a formulation of sequent calculus with a detailed comparison to {\bf LEci} discussed in \cite{pimentel2018,pimentel2019ecumenical}, as well as some comparison to the proposal due to Steffen Lewizka presented in \cite{Lewitzka2017}.

Needless to say, our paper is meant to be a starting point rather than offering our final words on this topic. Future directions to be pursued, beside fully addressing the question {\bf (Q2)} based on our new system, include (i) investigations into combination of subintuitionistic logic and classical logic, (ii) investigations into other ways of adding classical conditional on top of positive intuitionistic logic, such as those making use of Humberstone's constant $\Omega$ (cf. \cite{Humberstone2006,NikiOmoriOmega}), (iii) investigations into the combination via Beth semantics instead of Kripke semantics, and (iv) investigations into other ways to combine two conditionals, via other routes outlined in the introduction.

\bibliographystyle{eptcs}
\bibliography{biblio}

\end{document}